\documentclass[a4paper,10pt]{amsart}
\usepackage[utf8]{inputenc}
\usepackage[T1]{fontenc}
\usepackage{amsmath}
\usepackage{amsthm}
\usepackage{amsfonts}
\usepackage{mathrsfs}
\usepackage{color}
\usepackage{graphicx}

\newcommand{\N}{\mathbb{N}}
\newcommand{\C}{\mathbb{C}}
\newcommand{\R}{\mathbb{R}}
\newcommand{\Z}{\mathbb{Z}}

\newcommand{\ena}[1]{\mathrm{e}^{#1}}

\DeclareMathOperator{\sgn}{\mathrm{sgn}}
\newcommand{\dd}{\mathrm{d}}
\newcommand{\dx}{\frac{\dd}{\dd x}}
\newcommand{\dom}{\mathrm{Dom}}
\newcommand{\tr}{\mathrm{Tr}}
\newcommand{\G}[1]{\mathcal{G}_{#1}}
\DeclareMathOperator*{\nlim}{n-lim}
\DeclareMathOperator*{\esssupp}{ess\,supp}

\newtheorem{theorem}{Theorem}[section]
\newtheorem{lemma}[theorem]{Lemma}
\newtheorem{corollary}[theorem]{Corollary}
\newtheorem{proposition}[theorem]{Proposition}

\theoremstyle{definition}
\newtheorem{remark}[theorem]{Remark}

\title[Approximation of relativistic point interactions]{Approximation of one-dimensional relativistic point interactions by regular potentials revised}
\author{Mat\v{e}j Tu\v{s}ek}
\address{Department of Mathematics\\ Faculty of Nuclear Sciences and Physical Engineering \\ Czech Technical University in Prague\\ Trojanova 13\\ 120 00 Prague\\ Czechia}
\email{matej.tusek@fjfi.cvut.cz}
\date{\today}

\begin{document}

\begin{abstract}
 We show that the one-dimensional Dirac operator with quite general point interaction may be approximated in the norm resolvent sense by the Dirac operator with a scaled regular potential of the form $1/\varepsilon~h(x/\varepsilon)\otimes B$, where $B$ is a suitable $2\times 2$ matrix. Moreover,  we prove that the limit does not depend on the particular choice of $h$ as long as it integrates to a constant value. 
\end{abstract}

\maketitle

\section{Introduction}
The one-dimensional Dirac operator perturbed at one point is an important exactly solvable model of relativistic quantum mechanics. Mathematically, the perturbation is described by a boundary condition at the interaction point. Following \cite{BeDa_94} we will write it as 
\begin{equation} \label{eq:bc_int}
\psi(0+)=\Lambda\psi(0-),
\end{equation}
where $\psi$ is a two-component spinor and $\Lambda$ is from a four-parametric family 
of admissible matrices that lead to self-adjoint realizations of the Dirac operator and will be explicitly characterized later. Note that for convenience and without loss of generality the interaction point coincides with the origin.

The question how to approximate the point interactions by regular potentials is important for two major reasons. First, an approximation sequence may tell us much more about the nature of the point interaction rather then an abstract boundary condition. Second, various short range interactions may be well approximated by the point interactions and the latter are described by analytically solvable models. 
In the relativistic case, this question was addressed rigorously for the first time by \v{S}eba \cite{Se_89}. He focused exclusively on the so-called electrostatic and Lorentz scalar point interactions. Taking the former for example, he started with the Dirac operator with the  potential $1/\varepsilon ~h(x/\varepsilon)\otimes I$ for some $h\in L^1(\R;\R)$. Then he proved that, as $\varepsilon\to 0+$, this operator converges in the norm resolvent sense to the Dirac operator with the point interaction described by the boundary condition
\begin{equation} \label{eq:BC_Seba}
 \begin{split}
 &\frac{\eta}{2}(\psi_1(0+)+\psi_1(0-))=i(\psi_2(0+)-\psi_2(0-))\\
 &\frac{\eta}{2}(\psi_2(0+)+\psi_2(0-))=i(\psi_1(0+)-\psi_1(0-))\\
 \end{split}
\end{equation}
with 
\begin{equation} \label{eq:eta_int}
 \eta=\langle \sgn{(h)}|h|^{1/2},(1-K^2)^{-1}|h|^{1/2}\rangle_{L^2(\R)}.
\end{equation}
Here $K$ is the integral operator on $L^2(\R)$ with the kernel 
$$K(x,y)=\frac{i}{2}|h(x)|^{1/2}\sgn(x-y)\sgn(h(y))|h(y)|^{1/2}.$$

Next results are due to Hughes who found smooth local approximations to all types of the point interactions but only in the strong resolvent topology \cite{Hu_97,Hu_99}. She showed that if $h$ integrates to one then the Dirac operator with the potential $i/\varepsilon~ h(x/\varepsilon)\otimes\sigma_1 A$ converges to the Dirac operator with the point interaction described by the boundary condition \eqref{eq:bc_int} with $\Lambda=\exp{A}$. Of course, $A$ must be such that $\Lambda$ belongs to the admissible class.

The main purpose of this paper is to prove that the approximations found by Hughes converge also in the norm resolvent sense. Furthermore, explicit formulae for transition between $\Lambda$'s and  $A$'s,  which play important role in the proof, will be provided. Moreover, the quantity \eqref{eq:eta_int} will be calculated explicitly. In particular, it is constant once we fix the integral of $h$. Surprisingly, this remained unnoticed until now although the same quantity appears also in some very recent papers that deal with relativistic $\delta$-shell interaction in $\R^3$ and its approximation by short range potentials \cite{MaPi_18,MaPi_17}, where it plays a role of the interaction strength. Essentially the same quantity emerges also during the limiting procedure for other types of the point interactions, see Corollary \ref{cor:eta}. Consequently, the limit operator does not depend on particular choice of $h$ but only on the integral of $h$ and, of course, the matrix $A$. This was already observed by Hughes, although only in the strong resolvent topology. Therefore, one could have arrived to our findings about \eqref{eq:eta_int} indirectly just by combining the results from \cite{Se_89} and \cite{Hu_99}. It is also interesting that the spectrum of $K$ is $h$-independent (again with $\int_\R h$ fixed) and that one can find the eigenfunctions of $K$  explicitly. Finally, a new observation on the renormalization of the coupling constant is presented.

The paper is organized as follows. After reviewing some notation in Section \ref{sec:notation}, we will introduce the point interaction rigorously following \cite{FaGr_87} and \cite{BeDa_94} in Section \ref{sec:point_interaction}. Section \ref{sec:approximation} starts with a concise presentation of the results of Hughes. We also indicate some minor issues with her proof that we fix in a separate Appendix at the very end of the paper. Beside this, Appendix also contains very useful explicit formulae for passing from $\Lambda$ from the family of admissible matrices to all $A$'s such that $\exp{A}=\Lambda$ and $i\sigma_1 A$ is hermitian (and vice versa). Section \ref{sec:approximation} continues with the proof of the main result (Theorem \ref{theo:main}) on approximation of a huge family of the point interactions by regular potentials. This family comprises all special types of the point interactions that appear in literature (electrostatic, Lorenz scalar, $\delta$, and $\delta'$ interactions) and many more. The proof contains several spectral and other results for the operator $K$ (Propositions \ref{prop:K} and \ref{prop:renorm}, Lemma \ref{lem:zero_term}) that appear organically in the text exactly when needed. In Section \ref{sec:renormalization}, we briefly discuss the renormalization of the coupling constant during limiting procedure and show, providing explicit formulae, that for the point interactions from our family only a special type of the renormalization may occur.

\section{Notation} \label{sec:notation}
We will denote by $L^2(M;\mathscr{H})$ the Hilbert space of square-integrable functions on $M$ with values in a Hilbert space $\mathscr{H}$. If $\mathscr{H}=\C$ we will abbreviate $L^p(M;\mathscr{H})$ to $L^p(M)$. If, in addition, $M=\R$ we will sometimes write it just $L^p$ for brevity. We will use the symbol $\langle\cdot,\cdot\rangle$ exclusively for the dot product on $L^2(\R)$. When convenient we will identify $L^2(M)\otimes\C^2$ with $L^2(M;\C^2)$ and similarly for subspaces. We set $i\R:=\{i x\vert x\in\R\}$, $\pi\Z:=\{\pi k\vert\, k\in\Z\}$, $\N:=\{1,2,\ldots\}$, and $\N_0:=\N\cup\{0\}$. The Pauli matrices will be denoted by $(\sigma_i)_{i=1}^3$ and the Hilbert-Schmidt norm by $\|\cdot\|_2$. If $K$ is an integral operator then we will write $K(x,y)$ for its kernel.

\section{One dimensional Dirac operator with point interaction} \label{sec:point_interaction}
Let $M$ be a non-negative constant. Then it is well known that
\begin{equation*}
  H=-i\dx\otimes \sigma_1+M\otimes \sigma_3 
\end{equation*}
is self-adjoint on $W^{1,2}(\R)\otimes\C^2$  and 
$$\sigma(H)=\sigma_{ac}(H)=(-\infty,-M]\cup[M,+\infty).$$
Additional point interaction at $x=0$ may be introduced by firstly taking a symmetric restriction of $H$,
$$\dot{H}:=H|_{C_C^\infty(\R\setminus\{0\})\otimes\C^2}$$
and then looking for its self-adjoint extensions. The adjoint operator $\dot H^{*}$ acts as $H$ but on the domain
$$\dom(\dot H^*)=W^{1,2}(\R\setminus\{0\})\otimes\C^2$$
and has deficiency indices $(2,2)$. Therefore, every self-adjoint extension of $\dot{H}$ is described by some $2\times 2$ unitary matrix that is employed in the second von Neumann formula. A more convenient description of the extensions  in terms of a boundary condition at $x=0$ was used in \cite{BeDa_94}. 

Let us put
\begin{equation} \label{eq:Lambda}
 \Lambda=\ena{i\varphi}\begin{pmatrix}
                \alpha & i\beta\\
                -i\gamma & \delta
               \end{pmatrix}
\end{equation}
with $\varphi, \alpha,\, \beta,\, \gamma,\, \delta\in\R$ satisfying 
\begin{equation} \label{eq:lambda_det}
\alpha\delta-\beta\gamma=1.
\end{equation}
Note that one can only consider $\varphi\in[0,\pi)$ to obtain the whole class of the matrices. Then 
\begin{equation*}
 (H^\Lambda\psi)(x):=\Big(-i\dx\otimes\sigma_1+M\otimes\sigma_3\Big)\psi(x)\qquad (\forall x\in\R\setminus\{0\})
\end{equation*}
with 
\begin{equation*}
 \dom(H^\Lambda):=\{\psi\in W^{1,2}(\R\setminus\{0\})\otimes\C^2\vert\, \psi(0+)=\Lambda\psi(0-)\}
\end{equation*}
defines a self-adjoint extension of $\dot{H}$. In this way we obtain almost all self-adjoint extensions of $\dot{H}$, remaining being limiting cases, see \cite{BeDa_94}. There are also other ways how to parameterize the extensions that may be useful in different situations, cf.  \cite{Ge_Se_87,PaRi_14, AlVi_00,BeMaNe_08,CaMaPo_13}.

\section{Approximation by regular potentials} \label{sec:approximation}

\subsection{Approximation in the strong resolvent sense}

We start this section by looking closer at the results of Hughes \cite{Hu_97,Hu_99}. For every $H^\Lambda$, she found a family of self-adjoint Dirac operators with smooth potentials indexed by $\varepsilon>0$ that converges to $H^\Lambda$ \emph{in the strong resolvent sense} as $\varepsilon\to 0+$. More concretely, she considered  a family $\{h_\varepsilon\}_{\varepsilon>0}$ of smooth functions such that $h_\varepsilon\geq 0,\, \int_{\R}h_\varepsilon=1,\text{ and }\mathrm{supp}(h_\varepsilon)\subset[0,\varepsilon]$, and a matrix $A$ such that $\Lambda=\exp(A)$. Since $\Lambda$ is regular, there is always an $A$ with this property. The  approximating operators are then given by 
\begin{equation} \label{eq:approx}
\begin{split}
 &H^A_\varepsilon:=H+ih_\varepsilon(x)\otimes\sigma_1 A,\\
 &\dom(H^A_\varepsilon):=W^{1,2}(\R)\otimes\C^2.
\end{split}
\end{equation}

However, as far as I can see, there is a gap in her proof of \cite[Theorem 1]{Hu_99}. The first part of the proof  says, for a moment with  arbitrary matrices $A$ and $\Lambda$, that
 \begin{equation} \label{eq:Hu_1}
  \sigma_1 A\sigma_1=-A^*,
 \end{equation} 
i.e., $i\sigma_1 A$ is hermitian, if and only if $\Lambda=\exp(A)$ satisfies 
 \begin{equation} \label{eq:Hu_2}
  \Lambda^*\sigma_1\Lambda=\sigma_1,
 \end{equation}
in which case $\Lambda$ has the form \eqref{eq:Lambda}. As a counterexample, take 
\begin{equation*}
 \Lambda=\begin{pmatrix}
          \alpha & 0\\
          0 & \frac{1}{\alpha}
         \end{pmatrix},
\end{equation*}
with $\alpha>0$. It obeys \eqref{eq:Hu_2} (as any matrix of the form \eqref{eq:Lambda} does), but 
$$A:=\begin{pmatrix}
      \ln{\alpha}+i2\pi m & 0\\
      0 & -\ln{\alpha+i2\pi n}
     \end{pmatrix},
$$
with $m,n\in\Z$, does not fulfill \eqref{eq:Hu_1} unless $m=n$.

Since condition \eqref{eq:Hu_1} is crucial for the self-adjointness of $H^A_\varepsilon$, among all possible $A$'s one should only consider those that obey it. Therefore, the question is: \textit{Does for any $\Lambda$ of the form \eqref{eq:Lambda} exist $A$ such that $\Lambda=\exp(A)$ and \eqref{eq:Hu_1} holds true?} In Appendix, it is demonstrated that the answer is positive. In fact, there is always a whole one or two-parametric family of $A$'s with the desired properties, that can be written explicitly in terms of the entries of $\Lambda$.

\subsection{Approximation in the norm resolvent sense}
In this section we will show that, for a large subfamily of self-adjoint extensions of $\dot H$, the approximating operators $\{H^A_\varepsilon\}$ converge to $H^\Lambda$ (with $\Lambda=\exp(A)$) \emph{in the norm resolvent sense} as $\varepsilon\to 0+$. More concretely,
let $h\in L^1(\R;\R)\cap L^\infty(\R;\R)$ be such that
\begin{equation} \label{eq:norm_cond}
 \int_\R h(x)\dd x=1.
\end{equation}
For any $\varepsilon>0$, define
\begin{equation} \label{eq:h_def}
 h_\varepsilon(x):=\frac{1}{\varepsilon}h\Big(\frac{x}{\varepsilon}\Big)\quad \text{and} \quad \Theta_\varepsilon(x):=\int_{-\infty}^x h_\varepsilon(y)\dd y.
\end{equation}
Then $H_\varepsilon^A$ given by \eqref{eq:approx} is self-adjoint for any matrix $A$  obeying \eqref{eq:Hu_1} because it is a hermitian perturbation of $H$.
Note that $\Theta_\varepsilon$ is absolutely continuous, $\lim_{x\to -\infty}\Theta_\varepsilon(x)=0$, $\lim_{x\to +\infty}\Theta_\varepsilon(x)=1$, and $\lim_{\varepsilon\to 0}\Theta_\varepsilon=\Theta$, where $\Theta$ stands for the Heaviside function, point-wise on $\R\setminus\{0\}$. Furthermore,
let either
\begin{equation} \label{eq:A_1}
 A=i\varphi I, \quad\text{where }\varphi\in\R\setminus\{0,\,\lambda_k^{-1}\vert\, k\in\Z\}, 
\end{equation}
with 
\begin{equation} \label{eq:K_ev}
\lambda_k:=\frac{1}{(2k+1)\pi},
\end{equation}
or 
\begin{equation} \label{eq:A_2}
A:=\begin{pmatrix}
   a & i b\\
   i c & -a
   \end{pmatrix}, \quad\text{where }a,b,c\in\R\text{ and }\det{A}\notin\{\lambda_k^{-2}\vert\, k\in N_0\}.
\end{equation}
Note that \eqref{eq:Hu_1} is satisfied in the both cases.

Then we have
\begin{theorem} \label{theo:main}
 Let $h\in L^{1}(\R;\R)\cap L^\infty(\R;\R)$ obeys \eqref{eq:norm_cond}. Furthermore, let $H^A_\varepsilon$ be as in \eqref{eq:approx} with $h_\varepsilon$ given by \eqref{eq:h_def} and $A$ either of the form  \eqref{eq:A_1} or \eqref{eq:A_2}. Then
 \begin{equation*}
\lim_{\varepsilon\to 0+}\|(H^A_\varepsilon+i)^{-1}-(H^\Lambda+i)^{-1}\|=0  
 \end{equation*}
where $\Lambda=\exp(A)$.
\end{theorem}

Let us stress that for two special cases, namely $A=-i\sigma_1$ and $A=-\sigma_2$, the theorem has been already proved by \v{S}eba \cite{Se_89}. Also note that, since the norm resolvent convergence is stable with respect to adding a constant bounded perturbation \cite[Theorems IV.2.23~c) and IV.2.25]{Kato}, we will always assume that $M=0$ in the rest of this section, which is devoted to the proof Theorem \ref{theo:main}. We will focus on the more complex case with $A$ given by \eqref{eq:A_2}. The other choice of $A$ is briefly discussed in Remark \ref{rem:A_1}.

Our starting point is the resolvent formula by Kato \cite{Ka_66}, see also \cite{KoKu_66} or \cite{AlGe_04}, that was used by \v{S}eba too. Let  $\G{z}$ denote the integral kernel of the resolvent of $H$ at the point $z\in\C\setminus\sigma(H)=\C\setminus\R$,
\begin{equation*}
\G{z}(x,y)\equiv(H-z)^{-1}(x,y)=\frac{i}{2}(\zeta(z)I+\sgn(x-y)\sigma_1)\ena{iz\zeta(z)|x-y|}, 
\end{equation*}
where
\begin{equation*}
 \zeta(z):=\sgn(\Im z).
\end{equation*}
Then the resolvent formula yields
\begin{equation}
 (H^A_\varepsilon-z)^{-1}=(H-z)^{-1}-C_\varepsilon(1+Q_\varepsilon)^{-1}D_\varepsilon
\end{equation}
with $C_\varepsilon, Q_\varepsilon$, and $D_\varepsilon$ being the integral operators on $L^2(\R;\C^2)$ with the kernels
\begin{align*}
 &C_\varepsilon(x,y)=\G{z}(x,\varepsilon y)v(y),\\
 &Q_\varepsilon(x,y)=iu(x)\sigma_1 A\,\G{\varepsilon z}(x,y)v(y),\\
 &D_\varepsilon(x,y)=iu(x)\sigma_1 A\,\G{z}(\varepsilon x,y),
\end{align*}
where
\begin{equation*}
 u:=\sqrt{|h|}\quad \text{and} \quad v:=\sgn(h)\sqrt{|h|}.
\end{equation*}
One can check that $C_\varepsilon, D_\varepsilon$, and $Q_\varepsilon$ are Hilbert-Schmidt operators. Point-wise limits of their kernels, i.e.,
\begin{align*}
 &C(x,y):=\lim_{\varepsilon\to 0+}C_\varepsilon(x,y)=\G{z}(x,0)v(y),\\
 &Q(x,y):=\lim_{\varepsilon\to 0+}Q_\varepsilon(x,y)=iu(x)\sigma_1 A\,\G{0+}(x,y)v(y),\\
 &D(x,y):=\lim_{\varepsilon\to 0+}D_\varepsilon(x,y)=iu(x)\sigma_1 A\,\G{z}(0,y),
\end{align*}
are also associated with some Hilbert-Schmidt operators, say $C,Q$, and $D$, respectively. Here, $\G{0+}(x,y):=\lim_{\varepsilon\to 0+}\G{\varepsilon z}(x,y)$ remains $z$-dependent.

Using the Fubini and the dominated convergence theorems we infer that 
\begin{equation*}
 \lim_{\varepsilon\to 0 +}\|C_\varepsilon-C\|_{2}=0,\quad  \lim_{\varepsilon\to 0 +}\|Q_\varepsilon-Q\|_{2}=0,\quad  \lim_{\varepsilon\to 0 +}\|D_\varepsilon-D\|_{2}=0.
\end{equation*}
This implies the convergence in the uniform operator topology too.
Now, by the stability of bounded invertibility \cite[Theorem IV.1.16]{Kato}, if $-1\notin\sigma(Q)$ then for all $\varepsilon$ sufficiently small, $-1\notin\sigma(Q_\varepsilon)$ and
\begin{equation*}
 \nlim_{\varepsilon\to 0+} (I+Q_\varepsilon)^{-1}=(I+Q)^{-1}.
\end{equation*}
Here, $\nlim$ denotes the limit in the uniform operator topology. We conclude that
\begin{equation} \label{eq:res_lim}
 \nlim_{\varepsilon\to 0+} (H^A_\varepsilon-z)^{-1}=(H-z)^{-1}-C(I+Q)^{-1}D.
\end{equation} 
The right-hand side of \eqref{eq:res_lim} splits into a sum of the free resolvent and an operator of the rank two or smaller. Note that the latter is $z$-dependent too. This is exactly the decomposition that  appears in the Krein formula for the resolvent of $H^\Lambda$ \cite{BeDa_94}. To compare these two expressions one has to invert the operator $(I+Q)$.

Similarly as in \cite{Se_89}, we start by writing $Q=Q_1+Q_2$, where 
\begin{equation*}
 Q_1(x,y):=-\frac{\zeta(z)}{2}u(x)v(y)\sigma_1 A \quad \text{and} \quad  Q_2:=K\otimes \tilde{A}
\end{equation*}
with
\begin{equation*}
 K(x,y):=\frac{i}{2}u(x)\sgn(x-y)v(y)  \quad \text{and} \quad \tilde{A}:=i\sigma_1 A \sigma_1.
\end{equation*}
Clearly, the rank of $Q_1$ is at most two and $K$ is a Hilbert-Schmidt operator on $L^2(\R)$. We have
\begin{equation} \label{eq:I+Q_inv}
 (I+Q)^{-1}=(I+(I+Q_2)^{-1}Q_1)^{-1}(I+Q_2)^{-1},
\end{equation}
whenever the right-hand side makes sense. One can verify directly that
\begin{equation} \label{eq:I+Q_2_inv}
 (I+Q_2)^{-1}=(I-\nu^2 K^2)^{-1}\otimes I-K(I-\nu^2K^2)^{-1}\otimes\tilde{A},
\end{equation}
where 
$$\nu^2:=\det{A}=bc-a^2$$
is the same as in  Appendix, cf. \eqref{eq:nu_cond}. The candidate for $(I+Q_2)^{-1}$ was found by writing  it as a formal geometric series and taking the fact that, with our choice of $A$, $\tilde A^2=\nu^2 I$  into the account.
Note that the inverse of $(I-\nu^2 K^2)$ exists as a bounded operator if and only if $\nu^{-2}\notin\sigma(K^2)$. Let us look at the spectral properties of $K$ in detail.

\begin{proposition} \label{prop:K}
 $K$ is Hilbert-Schmidt with $\|K\|_2= \frac{1}{2}\|h\|_{L^1}$. For any $k\in\Z$,
 \begin{equation*}
  \psi_k(x):=\frac{1}{\sqrt{\|h\|_{L^1}}}\,u(x)\ena{i(2k+1)\pi \int_{-\infty}^x h(y)\dd y}
 \end{equation*}
 is a normalized eigenfunction of $K$ with the eigenvalue $\lambda_k$ given by \eqref{eq:K_ev}.
 Moreover, $\sigma(K)=\{0\}\cup\{\lambda_k\vert\, k\in\Z\}$ and $0\in\sigma_p(K)$ if and only if $h$ is not non-zero almost everywhere on $\R$. In the positive case, $0$ is an eigenvalue of infinite multiplicity.
 
 If $h\geq 0$ then $K$ is hermitian, therefore, its eigenfunctions form an orthonormal basis.
\end{proposition}

\begin{proof}
 The only non-trivial part is to find all eigenpairs. Let $\lambda\in\C$. If $K\psi=\lambda\psi$ then $\psi$ is necessarily factorized as $\psi(x)=u(x)f(x)$. Substituting this back into the eigenvalue equation we see that
 \begin{equation} \label{eq:int_eq}
  \frac{i}{2}\left(\int_{-\infty}^x h(y)f(y)\dd y-\int_x^{+\infty} h(y)f(y)\dd y\right)=\lambda f(x).
 \end{equation}
 Note that $hf\in L^1(\R)$ because we only deal with $\psi\in L^2(\R)$. Hence, for $\lambda\neq 0$, $f$ is absolutely continuous, and so we can take the derivative of \eqref{eq:int_eq}. We arrive at
 \begin{equation*} \label{eq:diff_eq}
  i h(x)f(x)=\lambda f'(x)
 \end{equation*}
 together with the boundary condition 
 \begin{equation} \label{eq:bc}
  \lim_{x\to -\infty}f(x)=-\lim_{x\to +\infty}f(x).
 \end{equation}
 We get $f(x)=C\exp(\frac{i}{\lambda}\int_{-\infty}^x h(y)\dd y)$. Imposing \eqref{eq:bc} on the solution we obtain 
 \begin{equation*}
  1=-\ena{\frac{i}{\lambda}}.
 \end{equation*}
Therefore, $\lambda$ must be as in \eqref{eq:K_ev}.

If $\lambda=0$ then the eigenvalue equation reduces to $hf=v\psi=0$. A non-trivial solution to this equation exists if and only if $h$ is not non-zero almost everywhere, i.e., $ \esssupp (h)\neq \R$. In the positive case, any non-zero $\psi\in L^2(\R\setminus \esssupp (h))$  is an eigenfunction.
\end{proof}

\begin{corollary}
 $\sigma(K^2)=\{0\}\cup\{\lambda_k^2\vert\, k\in\N_0\}$. In more detail, $\lambda_k^2,\, k\in\N_0,$ are twice degenerate eigenvalues of $K^2$ and $0$ is an eigenvalue of $K^2$ if and only if  $h$ is not non-zero almost everywhere on $\R$.
\end{corollary}

This is exactly why we assumed that $\det{A}\equiv\nu^2\neq\lambda_k^{-2}$, $k\in\N_0$, in \eqref{eq:A_2}. Note that if $\nu^2\leq 0$ then this is not a restriction at all. Now, in the the next step of inverting $(I+Q)$ we need to find $(I+(I+Q_2)^{-1}Q_1)^{-1}$. Take $g\otimes b\in L^2(\R)\otimes \C^2$ and look for a solution $f\otimes a$ in $L^2(\R)\otimes \C^2$ of the equation
\begin{equation} \label{eq:inv_eq}
 (I+(I+Q_2)^{-1}Q_1)f\otimes a=g\otimes b.
\end{equation}
Substituting for $Q_1$ and $Q_2$ we arrive at
\begin{multline} \label{eq:inverting}
 f\otimes a-\frac{\zeta(z)}{2}\langle v,f\rangle \left((I-\nu^2K^2)^{-1}u\otimes\sigma_1 A a-i K(I-\nu^2K^2)^{-1}u\otimes\sigma_1A^2 a\right)\\
 =g\otimes b.
\end{multline}
If we multiply the both sides of \eqref{eq:inverting} by $v$ with respect to the dot product on $L^2(\R)$ we get
\begin{equation} \label{eq:pert_term}
 \langle v,f\rangle a=\langle v,g\rangle\Big(I-\frac{\zeta(z)\eta}{2}\sigma_1 A\Big)^{-1}b,
\end{equation}
where we put
\begin{equation} \label{eq:eta}
 \eta:=\langle v, (I-\nu^2 K^2)^{-1} u\rangle
\end{equation}
and used the following observation (done already by \v{S}eba in \cite{Se_89} for the special case $\nu^2=\pm 1$).

\begin{lemma} \label{lem:zero_term}
 For all $\nu^2\in\C\setminus\{\lambda_k^{-2}\vert\, k\in\N_0\}$, we have 
 $$\langle v,K(I-\nu^2K^{2})^{-1}u\rangle=0.$$
\end{lemma}
\begin{proof}
Firstly, let $\nu^2\in\C:\ |\nu^2|<4/\|h\|_{L^1}^2$. Then $\|\nu^2K^2\|<1$ and we have
\begin{equation} \label{eq:geom_ser}
  \langle v,K(I-\nu^2K^{2})^{-1}u\rangle=\sum_{n=0}^{+\infty}\nu^{2n}\langle v, K^{2n+1}u\rangle.
\end{equation}
Let $S$ be the integral operator on $L^2(\R)$ with the kernel $S(x,y)=\frac{i}{2}\sgn(x-y)$. Then $K=uSv$ and, for any $f\in L^2(\R;\R)$, $\langle f,Sf\rangle=0$. Therefore, every term on the right-hand side of \eqref{eq:geom_ser} is zero because
\begin{equation} \label{eq:odd_powers}
\langle v, K^{2n+1}u\rangle=\langle v,(uSv)^{2n+1} u\rangle=\langle (hS)^n h, S (hS)^n h\rangle=0.
\end{equation}
Secondly, the function $F: \nu^2\mapsto \langle v,K(I-\nu^2K^{2})^{-1}u\rangle$ is analytic on $\C\setminus\{\lambda_k^{-2}\vert\, k\in\N_0\}$. Therefore, by the identity theorem, $F=0$ on $\C\setminus\{\lambda_k^{-2}\vert\, k\in\N_0\}$.
\end{proof}

Note that the inverse in \eqref{eq:pert_term} always exists because $\sigma(\sigma_1 A)\subset i\R$, $\zeta(z)\in\R$, and $\eta\in\R$. The latter will follow from Proposition \ref{prop:renorm}. Substituting \eqref{eq:pert_term} into \eqref{eq:inverting} we get
\begin{multline} \label{eq:inverse_cand}
 (I+(I+Q_2)^{-1}Q_1)^{-1}=I+\frac{\zeta(z)}{2}\langle v,\cdot\rangle\Big((I-\nu^2K^2)^{-1}u\\ \otimes\sigma_1 A \Big(I-\frac{\zeta(z)\eta}{2}\sigma_1 A\Big)^{-1}-iK(I-\nu^2K^2)^{-1}u\otimes\sigma_1 A^2 \Big(I-\frac{\zeta(z)\eta}{2}\sigma_1 A\Big)^{-1}\Big).
\end{multline}
Since we solved \eqref{eq:inv_eq} only on the decomposable vectors, one has to verify by a direct calculation that this is indeed the two-sided inverse. 
Now, inserting \eqref{eq:inverse_cand} together with \eqref{eq:I+Q_2_inv} into \eqref{eq:I+Q_inv} we finish the computation of 
$(I+Q)^{-1}$. Actually, we only need to know $C(I+Q)^{-1}D$ in \eqref{eq:res_lim}. For its integral kernel we have
\begin{equation*}
 (C(I+Q)^{-1}D)(x,y)=i\G{z}(x,0)\langle v, (I+Q)^{-1}u\rangle\sigma_1 A\G{z}(0,y).
\end{equation*}
Using Lemma \ref{lem:zero_term} we infer that
$$\langle v, (I+Q)^{-1}u\rangle=\eta I+\frac{\zeta(z)\eta^2}{2}\sigma_1 A\Big(I-\frac{\zeta(z)\eta}{2}\sigma_1 A\Big)^{-1}=\eta \Big(I-\frac{\zeta(z)\eta}{2}\sigma_1 A\Big)^{-1}.$$
After some more tedious computation we deduce that
\begin{equation*}
 (C(I+Q)^{-1}D)(x,y)=\G{z}(x,0)M(z)\G{z}(0,y),
\end{equation*}
where 
\begin{equation} \label{eq:Mz}
 M(z):=\frac{-1}{\frac{1}{2}\zeta(z)\eta(b+c)+i(1-\frac{1}{4}\eta^2\nu^2)}
 \begin{pmatrix}
 i\eta c+\frac{1}{2}\zeta(z)\eta^2\nu^2 & -a\eta\\
 a\eta & i\eta b+\frac{1}{2}\zeta(z)\eta^2\nu^2
 \end{pmatrix}.
\end{equation}

It remains to evaluate $\eta$ as a function of $\nu$.
\begin{proposition} \label{prop:renorm}
 For all $\nu\in\C\setminus\{0,\,\lambda_k^{-1}\vert\, k\in\Z\}$,
 \begin{equation} \label{eq:renorm}
  \eta=\langle v, (I-\nu^2 K^2)^{-1} u\rangle=\frac{2}{\nu}\tan\frac{\nu}{2}.
 \end{equation}
\end{proposition}

\begin{proof}
As in the proof of Lemma \ref{lem:zero_term} we start with the case $|\nu^2|<4/\|h\|_{L^1}^2$. Then
\begin{equation} \label{eq:geom_ser2}
\langle v, (I-\nu^2 K^2)^{-1} u\rangle=\sum_{n=0}^{+\infty}\nu^{2n}\langle v, K^{2n}u\rangle,
\end{equation}
where
\begin{multline*}
 \langle v,K^{2n} u\rangle=\frac{(-1)^n}{2^{2n}}\int_{\R^{2n+1}}h(x_1)\sgn(x_1-x_2)h(x_2)\sgn(x_2-x_3)h(x_3)\ldots\\
 \sgn(x_{2n}-x_{2n+1})h(x_{2n+1})\dd x_1\ldots\dd x_{2n+1}=\\
 \frac{(-1)^n}{(2n+1)!\, 2^{2n}}\int_{\R^{2n+1}}\prod_{j=1}^{2n+1}h(x_j)\sum_{\tau\in S_{2n+1}}\prod_{l=1}^{2n}\sgn(x_{\tau(l)}-x_{\tau(l+1)})\dd x_1\ldots\dd x_{2n+1}.\\
\end{multline*}
If we define
\begin{equation*}
 C_n(x_1,x_2,\ldots,x_{2n+1}):=\sum_{\tau\in S_{2n+1}}\prod_{l=1}^{2n}\sgn(x_{\tau(l)}-x_{\tau(l+1)})
\end{equation*}
then we observe that that $C_n$ is a constant (which will be denoted by the same letter) almost everywhere on $\R^{2n+1}$. Therefore, using \eqref{eq:norm_cond} we get
$$\langle v,K^{2n} u\rangle=\frac{(-1)^n C_n}{(2n+1)!\, 2^{2n}}.$$

To find the explicit value of $C_n$ is an exercise in combinatorics. Alternatively, since we have just showed that the quantity $\langle v,K^{2n} u\rangle$ does not depend on a particular choice of $h$, one can use the result of \v{S}eba \cite{Se_89}, where it is calculated for  $h=\chi_{(0,1)}$. This way we obtain
$$\langle v,K^{2n} u\rangle= \frac{(-1)^n 4 (2^{2n+2}-1)}{(2n+2)!}B_{2n+2},$$
where $B_{2n+2}$ stands for the $(2n+2)$th Bernoulli number. Substituting this into \eqref{eq:geom_ser2} we arrive at \eqref{eq:renorm}.

By analyticity of the both sides, one can extend \eqref{eq:renorm} on $\C\setminus\{0,\,\lambda_k^{-1}\vert\, k\in\Z\}$ using the identity theorem.
\end{proof}

\begin{corollary} \label{cor:eta}
 If we started with $h$ that does not integrate to one then
 \begin{equation} \label{eq:eta_nonnorm}
  \eta=\langle v, (I-\nu^2 K^2)^{-1} u\rangle=\frac{2}{\nu}\tan\frac{\nu\int_\R h(x)\dd x}{2}. 
 \end{equation}
\end{corollary}

Let us now recall the formula for the resolvent of $H^\Lambda$ as was derived in \cite{BeDa_94},
\begin{equation*}
 (H^\Lambda-z)^{-1}(x,y)=\G{z}(x,y)-\G{z}(x,0)M^\Lambda(z)\G{z}(0,y),
\end{equation*}
where
\begin{multline*}
 M^\Lambda(z):=\frac{1}{\zeta(z)(\beta-\gamma)+i(\alpha+\delta)}\\
  \begin{pmatrix}
   2i\gamma+\zeta(z)(\alpha+\delta-2\cos\varphi) & \alpha-\delta-2i\sin\varphi\\
   \delta-\alpha-2i\sin\varphi & -2i\beta+\zeta(z)(\alpha+\delta-2\cos\varphi)                                                             
  \end{pmatrix}.
\end{multline*}
To conclude the proof of Theorem \ref{theo:main} we have to express $M^\Lambda(z)$ in terms of the entries of $A$ and then show that $M^\Lambda(z)=M(z)$. Following Appendix we should distinguish three cases, $\nu\in\R\setminus\pi\Z,\, \nu\in\pi\Z,$ and $\nu\in i\R\setminus\{0\}$. We will focus on the first case, the latter may be treated similarly. In the second case, if $\nu=2k\pi$ with $k\in\Z\setminus\{0\}$ then $\Lambda=I$, $\eta=0$, and $M(z)=0$, i.e., there is no point interaction in the limit as $\varepsilon\to 0+$. If $\nu=(2k+1)\pi=\lambda_k^{-1}$ with $k\in\Z$ then  we cannot use our approach. If $\nu=0$ then $\eta=1$ and one can proceed similarly as in the first case. In particular, the corresponding $M^\Lambda(z)$ is the same as the limit of \eqref{eq:MzLambda} as $\nu\to 0$.

Let $\nu\in\R\setminus\pi\Z$. Starting with $A$ of the form \eqref{eq:A_2} we can use \eqref{eq:Im_a}, \eqref{eq:cos_mu}, and \eqref{eq:sin_mu} with $n=0$ to calculate that
\begin{equation*}
 \alpha=\cos\nu+\frac{\sin\nu}{\nu}a,\, \beta=\frac{\sin\nu}{\nu}b,\, \gamma=-\frac{\sin\nu}{\nu}c,\,
 \delta=\cos\nu-\frac{\sin\nu}{\nu}a.
\end{equation*}
After some algebra we obtain
\begin{multline} \label{eq:MzLambda}
 M^\Lambda(z)=\frac{-1}{\zeta(z)\frac{\tan{\frac{\nu}{2}}}{\nu}(b+c)+i(1-\tan^2{\frac{\nu}{2}})}\\
   \begin{pmatrix}
      i\frac{2\tan{\frac{\nu}{2}}}{\nu}c +2 \zeta(z)\tan^2{\frac{\nu}{2}}& -\frac{2\tan{\frac{\nu}{2}}}{\nu}a \\
      \frac{2\tan{\frac{\nu}{2}}}{\nu}a & i\frac{2\tan{\frac{\nu}{2}}}{\nu}b+2\zeta(z)\tan^2{\frac{\nu}{2}}
   \end{pmatrix}.
\end{multline}
If we substitute \eqref{eq:renorm} into \eqref{eq:Mz} we get exactly the same matrix.

\begin{remark} \label{rem:A_1}
 If we start with $A$ of the form \eqref{eq:A_1} then 
 $$(I+Q_2)^{-1}=(I-\varphi K)^{-1}\otimes I.$$
 Since $\varphi\notin\{\lambda_k^{-1}\vert\, k\in\Z\}$, the inverse on the right-hand side always exists due to Proposition \ref{prop:K}. Next, we have
 \begin{equation} \label{eq:invAlt}
  (I+(I+Q_2)^{-1}Q_1)^{-1}=I+\frac{i\zeta(z)\varphi}{2}\langle v,\cdot\rangle(I-\varphi K)^{-1}u\otimes\sigma_1\Big(I-\frac{i\zeta(z)\varphi\tilde{\eta}}{2}\sigma_1\Big)^{-1}
 \end{equation}
 with
 \begin{equation*}
  \tilde{\eta}:=\langle v,(I-\varphi K)^{-1}u\rangle.
 \end{equation*}
 The inverse matrix on the very right of \eqref{eq:invAlt} exists because $\tilde{\eta}\in\R$. In fact, employing \eqref{eq:odd_powers} and Proposition \ref{prop:renorm} one can show that
 \begin{equation*}
  \tilde{\eta}=\langle v,(I-\varphi^2K^2)^{-1}u\rangle=\frac{2}{\varphi}\tan{\frac{\varphi}{2}}.
 \end{equation*}
 Finally, we obtain
 \begin{equation*}
  \langle v, (I+Q)^{-1}u\rangle=\tilde{\eta}\Big(I-\frac{i\zeta(z)\varphi\tilde{\eta}}{2}\sigma_1\Big)^{-1}
 \end{equation*}
 and, consequently,
 \begin{equation*}
 M(z)=-\sin\varphi\begin{pmatrix}
                  i\zeta(z)\tan\frac{\varphi}{2} & 1\\
                  1 & i\zeta(z)\tan\frac{\varphi}{2}
                 \end{pmatrix}.
 \end{equation*}
 This is exactly $M^\Lambda(z)$ with $\Lambda=\exp(A)=\exp(i\varphi)I$.
\end{remark}

\section{Renormalization of $\delta$ potentials} \label{sec:renormalization}
Speaking about renormalization we have to specify how we measure the strength of point interactions in the the first place. Let us look at the electrostatic point interaction for example. The boundary condition \eqref{eq:BC_Seba} may be rewritten as \eqref{eq:bc_int} with
\begin{equation*}
 \Lambda=\frac{1}{1+\frac{\eta^2}{4}}\begin{pmatrix}
          1-\frac{\eta^2}{4} & -i\eta\\
          -i\eta & 1-\frac{\eta^2}{4}
          \end{pmatrix}.
\end{equation*}
Given appropriate approximations of the form \eqref{eq:approx} with $h$ that does not necessarily fulfills \eqref{eq:norm_cond}, $\eta$ is given by \eqref{eq:eta_nonnorm} with $\nu=1$, i.e.,
$ \eta=2\tan(\int_\R h/2)$.
This scales non-homogeneously with respect to $h$. However, if we decide to measure the interaction strength by a new parameter  $\theta$ that obeys $\eta=2\tan(\theta/2)$ then
\begin{equation*}
 \Lambda=\begin{pmatrix}
         \cos{\theta} & -i\sin{\theta}\\
         -i\sin{\theta} & \cos{\theta}
         \end{pmatrix}
\end{equation*}
and $\theta=\int_\R h$ is obviously linear in $h$.  The reason why $\eta$ (and not $\theta$) is referred to as \emph{the coupling constant} for electrostatic potential is that by formal integration of the eigenvalue equation with potential $\eta \delta$ one arrives at \eqref{eq:BC_Seba}. During this procedure one extends $\delta$-distribution to the functions with discontinuity at $x=0$ by putting $\psi(0):=(\psi(0+)+\psi(0-))/2$. Mathematically, this was justified in \cite{BoKu_98}.

For a general point interaction, it was deduced by Hughes \cite{Hu_99} that 
\begin{equation*}
 H^\Lambda\psi=-i\dx\otimes \sigma_1 \psi+M\otimes \sigma_3 \psi+2i\otimes\sigma_1(\Lambda-I)(\Lambda+I)^{-1}\psi(0)\delta
\end{equation*}
in the sense of distributions. Since, in \eqref{eq:approx}, $\lim_{\varepsilon\to 0+}h_\varepsilon=\delta$ in the sense of distributions, it is reasonable to state that renormalization of the coupling constant does not occur if and only if
\begin{equation} \label{eq:renorm_cond}
 A=2(\exp(A)-I)(\exp(A)+I)^{-1},
\end{equation}
provided that the inverse exists. See \cite{Hu_99} for some necessary and sufficient conditions on $A$ under which \eqref{eq:renorm_cond} is fulfilled.

Now, using Appendix we can deduce that the right-hand side of \eqref{eq:renorm_cond} equals
\begin{equation*}
 W^A:=\frac{2}{\nu(\cos{\nu}+\cos\varphi)}\begin{pmatrix}
                                      \Re{a}\sin\nu+i\nu\sin\varphi & ib\sin\nu\\
                                      ic\sin\nu & -\Re{a}\sin\nu+i\nu\sin\varphi
                                     \end{pmatrix},
\end{equation*}
where $\varphi=\Im{a}$ and $\nu$ is given by \eqref{eq:nu_cond}. For $\nu=0$, one can just send $\nu$ to zero in the formula.
Firstly, consider $A$ given by \eqref{eq:A_1}. A straightforward calculation yields
\begin{equation*}
W^A=\frac{2}{\varphi}\tan\frac{\varphi}{2} A. 
\end{equation*}
Next, let $A$ be as in \eqref{eq:A_2}, i.e., $a\in\R$. Then
\begin{equation*}
W^A=\frac{2}{\nu}\tan\frac{\nu}{2}A.
\end{equation*}
Therefore, for both types of the studied matrices, $W^A$ is a multiple of $A$. However, the renormalization occurs unless $\varphi/2=\tan(\varphi/2)$ or $\nu/2=\tan(\nu/2)$, respectively. 
Finally, we conclude that if the potential in \eqref{eq:approx} is multiplied by $t\in\R$ then the strength of the point interaction scales as $2/\varphi\tan(t\varphi/2)$ and $2/\nu\tan(t\nu/2)$, respectively.

\section*{Appendix}
Here we prove that \emph{for any $\Lambda$ of the form \eqref{eq:Lambda}, there exists $A$ such that $\Lambda=\exp(A)$ and \eqref{eq:Hu_1} holds true.} Moreover, we find explicit formulae for \emph{all} $A$'s with these properties in terms of the entries of $\Lambda$.

Firstly, note that $A$ obeys \eqref{eq:Hu_1} if and only if
\begin{equation} \label{eq:A_matrix}
 A=\begin{pmatrix}
    a & i b\\
    i c & -\bar a
   \end{pmatrix},\qquad\text{where } a\in\C\text{ and }b, c\in\R.
\end{equation}

Secondly, for an arbitrary $2\times 2$ matrix $A$ such that
\begin{equation*}
 \nu:=\sqrt{\det A-\left(\frac{\tr A}{2}\right)^2}\neq 0
\end{equation*}
we have
\begin{equation} \label{eq:exp}
 \exp(A)=\exp\left(\frac{\tr A}{2}\right)\left(\cos\nu~I+\frac{\sin\nu}{\nu}\left(A-\frac{\tr A}{2}I\right)\right).
\end{equation}
Note that both branches of the square-root  lead to the same formula.
If $\nu=0$ then we write $1$ instead of $\sin \nu/\nu$ in \eqref{eq:exp}. The formula can be derived by splitting $A$ into a sum of a traceless matrix and a multiple of identity. The decomposition is unique and its two parts commute. The exponential of a multiple of identity is trivial to calculate, the exponential of a traceless matrix can be also calculated directly from the Taylor series because its square is diagonal and, therefore, we can sum up the odd and even terms separately. Note that the right-hand side of \eqref{eq:exp} again decomposes into a unique sum of  a multiple of identity and a traceless matrix. Essentially the same formula  appears in \cite{BeSo_93} but with a different proof.

Now, for $A$ of the form \eqref{eq:A_matrix}, $\tr A/2= i\Im a$, 
\begin{equation*}
 A-\frac{\tr A}{2}I=\begin{pmatrix}
                     \Re a & i b\\
                     i c   & -\Re a
                    \end{pmatrix},
\end{equation*}
and 
\begin{equation} \label{eq:nu_cond}
\nu^2=bc-(\Re a)^2\in\R.
\end{equation}
Hence, either $\nu^2\geq 0$ (equivalently, $\nu\in\R$) or $\nu^2<0$ (equivalently, $\nu=i g$ with $g\in\R\setminus\{0\}$, and consequently, $\cos \nu=\cosh g$ and $\sin \nu=i\sinh g$). Now, we write
\begin{equation} \label{eq:Lambda_dec}
\Lambda=\ena{i\varphi}\left(\frac{\alpha+\delta}{2}~I+\begin{pmatrix}
                                                       \frac{\alpha-\delta}{2} & i\beta\\
                                                       -i\gamma & \frac{\delta-\alpha}{2}
                                                      \end{pmatrix}
\right).
\end{equation}
and compare it with \eqref{eq:exp}.

\subsubsection*{The case $(\alpha+\delta)/2\in(-1,1)$} 
In this case $\nu\in\R\setminus\pi\Z$ and we have
\begin{align}
 & \Im a=\varphi +n\pi,\quad n\in\Z, \label{eq:Im_a}\\
 & (-1)^n\cos \nu=\frac{\alpha+\delta}{2}, \label{eq:cos_mu}\\
 & (-1)^n\frac{\sin \nu}{\nu}\Re a=\frac{\alpha-\delta}{2},\quad (-1)^n\frac{\sin \nu}{\nu} b=\beta,\quad  (-1)^n\frac{\sin \nu}{\nu} c=-\gamma. \label{eq:sin_mu}
\end{align}
Note that $\eqref{eq:Im_a}$ follows also directly from the formula $\det(\Lambda)=\exp(\tr A)$.
Clearly, \eqref{eq:cos_mu} has exactly two solutions in any half-closed interval of length $2\pi.$ With any of these solutions we calculate $\Re a,\, b$, and $c$ from \eqref{eq:sin_mu}. Since
\begin{equation*}
bc-(\Re a)^2=\frac{\nu^2}{\sin^2 \nu}\left(-\beta\gamma-\left(\frac{\alpha-\delta}{2}\right)^2\right)=\frac{\nu^2}{\sin^2 \nu}\left(1-\left(\frac{\alpha+\delta}{2}\right)^2\right)=\nu^2,
\end{equation*}
the equations \eqref{eq:cos_mu} and \eqref{eq:sin_mu} are compatible with \eqref{eq:nu_cond}. Here,
we used \eqref{eq:lambda_det} and \eqref{eq:cos_mu} in the second and the third equality, respectively.

\subsubsection*{The case $(\alpha+\delta)/2=\pm 1$} 
In this case, $\nu=m\pi$ with $m\in\Z$. Since we can absorb the sign of $(\alpha+\delta)/2$ into the phase factor of $\Lambda$, we may only focus on the case $(\alpha+\delta)/2=1$. Firstly, for $m\neq 0$, $\exp(A)=\exp{(i\Im a)}(-1)^m I$. Therefore, if $\Lambda=\exp{(i\varphi)} I$ then $\exp{(A)}=\Lambda$ for arbitrary $A$ that satisfies
\begin{equation*}
 bc-(\Re{a})^2=(m\pi)^2\quad\text{and}\quad\Im{a}=\varphi+m\pi+2n\pi 
\end{equation*}
with any $n\in\Z$. If $\Lambda$ is not a multiple of identity then there is no solution $A$ with $m\neq 0$.

Secondly, let $m=\nu=0$. Then $n$ in \eqref{eq:Im_a} must be even so that \eqref{eq:cos_mu} is fulfilled. Moreover, we have
\begin{equation*}
 \Re a=\frac{\alpha-\delta}{2},\quad  b=\beta,\quad  c=-\gamma
\end{equation*}
instead of \eqref{eq:sin_mu}. This yields a valid solution because it is  compatible with \eqref{eq:nu_cond}.

\subsubsection*{The case $(\alpha+\delta)/2\in\R\setminus[-1,1]$}
In this case $\nu=ig$, with $g\in\R\setminus\{0\}$, and we have \eqref{eq:Im_a} together with
\begin{align}
 & (-1)^n\cosh g=\frac{\alpha+\delta}{2}, \label{eq:cosh_g}\\
 & (-1)^n\frac{\sinh g}{g}\Re a=\frac{\alpha-\delta}{2},\quad (-1)^n\frac{\sinh g}{g}b=\beta,\quad (-1)^n\frac{\sinh g}{g}c=-\gamma. \label{eq:sinh_g}
\end{align}
Clearly,  $n$ must be even for $(\alpha+\delta)>0$ and odd in the other case. Now, the two solutions of \eqref{eq:cosh_g} differ only by sign, so if we substitute either of them into \eqref{eq:sinh_g} we arrive at the same formulae for $\Re a,\, b$, and $c$. Therefore, the only ambiguity remains in $\Im a$ that is given modulo $2\pi$. The compatibility of \eqref{eq:cosh_g} and \eqref{eq:sinh_g} with \eqref{eq:nu_cond} follows from \eqref{eq:lambda_det} again.

\begin{remark}
Using \eqref{eq:exp} one can infer that starting with $A$ given by \eqref{eq:A_matrix} one always ends up with $\Lambda=\exp(A)$ of the form \eqref{eq:Lambda}.
\end{remark}

\section*{Acknowledgments}
This work was supported by the project CZ.02.1.01/0.0/0.0/16\_019/0000778 from the European Regional Development Fund and by the grant No. 17-01706S of the Czech Science Foundation (GA\v{C}R).
The author wishes to express his thanks to V. Lotoreichik for many fruitful discussions and to D. Krej\v{c}i\v{r}\'{i}k for pointing out a useful reference.

\end{document}